\documentclass[a4paper,UKenglish,runningheads,11pt]{llncs}

\usepackage{tabularx,booktabs,multirow,delarray,array}
\usepackage{graphicx,amssymb,amsmath,amssymb,mathtools,amsfonts}
\usepackage{enumerate}
\usepackage[ruled,vlined,linesnumbered]{algorithm2e}
\usepackage{wrapfig}
\usepackage{latexsym}
\usepackage{lineno}
\usepackage{hyperref}
\usepackage{subcaption}
\usepackage{float}
\usepackage{url}
\usepackage{upgreek}
\usepackage{textgreek}
\usepackage{caption}
\usepackage{subcaption}
\usepackage{enumitem}
\usepackage{fullpage}

\newenvironment{proof}{\par\noindent{\bf Proof:}}{\mbox{}\hfill$\qed$\\}

\newlength{\bibitemsep}
\newlength{\bibparskip}
\let\oldthebibliography\thebibliography
\renewcommand\thebibliography[1]{
  \oldthebibliography{#1}
  \setlength{\parskip}{\bibitemsep}
  \setlength{\itemsep}{\bibparskip}
}

\newcommand{\ignore}[1]{ }

\newcounter{rem}
\def\etal{\textsl{et~al.}}

\def\qed{\hbox{\rlap{$\sqcap$}$\sqcup$}}

\begin{document}

\title{A Recursive Algorithm for Routing amid Convex Polygonal Obstacles}
\titlerunning{Routing in Convex Polygonal Domains}

\author{
Siddharth Gaur\inst{1}
\and
R. Inkulu\inst{1}
}

\institute{
Department of Computer Science and Engineering\\
Indian Institute of Technology Guwahati\\
\email{\{sgaur,rinkulu\}@iitg.ac.in}
}

\authorrunning{S. Gaur and R. Inkulu}

\maketitle

\pagenumbering{arabic}
\setcounter{page}{1}

\begin{abstract}
Given a polygonal domain $\cal P$ comprising $h$ pairwise disjoint convex polygonal obstacles in the plane, together defined with $n$ vertices, this paper presents an algorithm to preprocess $\cal P$ to compute routing tables at the vertices of $\cal P$ so that a data packet from any vertex of $\cal P$ is routed to any other vertex belonging to $\cal P$.
At every vertex $v$ of $\cal P$ along the routing path, until the packet reaches its destination, the next hop is determined using the routing tables at $v$ and the information stored in the packet header. 
In $O(n^2(\lg{n}))$ time, our preprocessing algorithm assigns a unique label of size $O(\sqrt{h} (\lg{h}) \lg{n})$ to each vertex of $\cal P$ and computes routing tables of size $O(h\lg{n} + \sqrt{h}(\lg{h})(\min((\frac{1}{\epsilon})^{O( \lg {\alpha})},n))$ $\lg {n})$ at each vertex of $\cal P$.
The routing path output has a $(7 + \epsilon)(\lg{h})$ multiplicative stretch.  
Here, $\epsilon > 0$ is an input parameter and $\alpha > 1$ is a geometric parameter.
\end{abstract}


\section{Introduction}
\label{sect:intro}

A {\it polygonal domain} $\cal P$ in $\mathbb{R}^2$ is a finite set of pairwise disjoint simple polygonal obstacles.
The free space $\cal{F(P)}$ of $\cal P$ is the plane that contains obstacles in $\cal P$, excluding the interiors of those obstacles. 
The polygonal domain $\cal P$ is a {\it convex polygonal domain} if every polygon in $\cal P$ is convex.
The {\it routing problem} for $\cal P$ seeks to find a route in $\cal{F(P)}$ for a data packet from any source vertex $s \in {\cal P}$ to any destination vertex $t \in {\cal P}$.
At every vertex $v$ along that routing path, the next hop on that path is determined using the information stored in the routing tables at $v$ and the data in the header of the packet. 
These routing tables at the vertices of $\cal P$ are computed in the preprocessing phase of the algorithm.
In the routing phase, packets are routed from any vertex of $\cal P$ to any other vertex.
The typical objectives of a routing problem include minimizing the following:
(i) the maximum ratio of the length of the routing path to the length of the shortest path between any two vertices of $\cal P$ in $\cal{F(P)}$,
(ii) the space required for routing tables,
(iii) the number of hops along the routing path, 
(iv) the preprocessing time to construct routing tables, and
(v) the size of the packet header.

Computing a shortest path between two given points in a polygonal domain is a fundamental problem in computational geometry.
Naturally, the routing problem is closely related to computing optimal or approximate shortest paths in polygonal domains.
This problem is primarily studied using two approaches.
In one approach, by constructing a graph in $\cal{F(P)}$, called a visibility graph, whose nodes are the vertices of $\cal P$ and edges are the line segments between mutually visible vertices, a shortest path of interest is determined by computing a shortest path in that visibility graph \cite{journals/siamcomp/KapoorM00,journals/dcg/KapoorMM97,journals/ipl/Welzl85}.
Several algorithms for computing visibility graphs are detailed in the monograph by Ghosh~\cite{books/visalgo/skghosh2007}.
The second approach uses the continuous Dijkstra or wavefront propagation technique \cite{journals/siamcomp/HershbergerS99,journals/corr/InkuluKM10,conf/stoc/Kapoor99,conf/socg/Kapoor88,journals/ijcga/Mitchell96}.
A Dijkstra wavefront is initiated at the source and is expanded in $\cal{F(P)}$ until it strikes the destination.
Using this approach, Hershberger and Suri~\cite{journals/siamcomp/HershbergerS99} devised an $O(n\lg{n})$ time algorithm.
Subsequently, Inkulu, Kapoor, and Maheshwari~\cite{journals/corr/InkuluKM10} extended Kapoor’s algorithm~\cite{conf/stoc/Kapoor99} and devised an algorithm with $O(n+h((\lg{h})^{\delta}+(\lg{n})(\lg{h})))$ time complexity.
Here, $\delta$ is a small positive constant (resulting from the time to triangulate $\cal{F(P)}$ with the algorithm in Bar-Yehuda and Chazelle~\cite{journals/ijcga/Bar-YehudaC94}).
Finding shortest paths in polygonal domains is extensively studied \cite{journals/algorithmica/AsanoAGHI86,conf/soda/ChiangM99,journals/algorithmica/GuibasHLST87,journals/comgeo/InkuluK09a,conf/isaac/InkuluKap19,conf/socg/OverWelzl88,journals/siamcomp/SharirS86,journals/jacm/StorerR94,conf/isaac/InkuluKap19}.
Canny and Reif~\cite{conf/focs/CannyR87} showed that computing a shortest path among polyhedral obstacles in $\mathbb{R}^3$ is NP-hard.
Mitchell~\cite{coll/hb/Mitch17} surveys the shortest path algorithms in geometric domains.
The fundamental data structures and algorithms of computational geometry are detailed in the popular textbooks by Preparata and Shamos~\cite{books/compgeom/prep1985} and de~Berg~\etal~\cite{books/compgeom/deberg2008}.

The stretch in geometric spanner networks is related to the stretch obtained in routing algorithms.
Given a graph $G$, a subgraph $H$ of $G$ is a $t$-spanner of $G$ for $t \ge 1$ whenever for all pairs of vertices $u$ and $w$ in $G$, $d_G(u, w) \le d_H(u, w) \le t \cdot d_G(u, w)$.
A spanner network of a complete graph embedded in a geometric domain is called a geometric spanner network.
The spanner networks for polygonal domains are studied in \cite{conf/esa/ArikatiCCDSZ96,conf/soda/Chen95,conf/socg/ClarksonKV87}.
The monograph \cite{books/compgeom/narsmid2007} by Narasimhan and Smid and an article by Bose and Smid~\cite{journals/cgta/BoseSmid13}, provide a comprehensive survey of results on geometric spanners.

A graph is a geometric graph if it is embedded in the plane.
The geometric graphs Yao-graphs~\cite{journals/siamcomp/Yao82} and \textTheta-graphs~\cite{conf/stoc/Clarkson87} are commonly used in computing geometric spanner networks. 
For geometric graphs, local routing algorithms determine the next hop based on the location of $t$, the location of the vertex $v$ at which the packet is residing, and the neighbours of $v$ in the graph.
Significant works for routing in geometric graphs include \cite{conf/isaac/BoseKRV17,journals/tcs/BoseMorin04,conf/distrcomp/HassPeleg00,conf/cccg/Kranakis99}.
The routing algorithms on Yao-graphs, \textTheta-graphs, and half-$\textnormal{\textTheta}_6$-graphs are considered in Bose~\etal~\cite{journals/siamjc/BoseFRV15} and Bose~\etal~\cite{journals/jocg/BoseCD20}.
The $\textnormal{\textTheta}_6$-graph is a \textTheta-graph that has six cones.
Given a set of $n$ points, Peleg~\etal~\cite{conf/distrcomp/HassPeleg00} constructed a sparse network with a multiplicative stretch of $\frac{16\sqrt{2}}{\cos\theta-\sin\theta}$ and with size of the routing tables $O((\lg{D}) (\lg{n}))$ per vertex, where $D$ is the diameter of the point set.

Given a polygonal domain defined with $n$ vertices and $h$ simple polygonal obstacles, 
Banyassady~\etal~\cite{journals/cgta/BanyaCKMR20} preprocess the polygonal domain in $O(n^2\lg{n} + \frac{n}{\epsilon})$ time and construct routing tables of size $O((\frac{1}{\epsilon}+h)\lg {n})$ per vertex, and the routing path computed has $1+\epsilon$ multiplicative stretch, for any $\epsilon > 0$. 
When the polygonal domain consists of convex polygonal obstacles, Inkulu and Kumar~\cite{journals/ijfcs/InkuluKum24} preprocess the polygonal domain in $O(n+\frac {h^{3}}{\epsilon^{2}}polylog(\frac{h}{\epsilon}))$ time for $\epsilon < 1$, and compute routing tables of size $O(\frac {h^{2}}{\epsilon} polylog(\frac{h}{\epsilon}))$ bits per vertex, and the routing path has $1+\epsilon$ multiplicative stretch and $2 k \ell$ additive stretch.
Gaur and Inkulu~\cite{journals/joco/GaurInkulu25}, a simple polygon defined with $n$ vertices is preprocessed with a divide-and-conquer algorithm in $O(n (1+\frac{1}{\epsilon}) (\lg{n})^3)$ time to construct routing tables at all the vertices, together of size $O(n+\frac{n}{\epsilon}(\lg{n})^3)$, facilitating the computation of a routing path with a multiplicative stretch $(2+\epsilon)\lg{n}$, for any $\epsilon > 0$.

The routing problem is also popular in abstract graphs.
In the case of graphs, a packet needs to be routed from any source vertex of the input graph $G$ to the destination vertex of $G$ along a path in $G$. 
The preprocessing algorithm computes routing tables at every vertex of $G$.
A naive preprocessing algorithm could compute all-pairs shortest paths in $G$. 
For a vertex $u'$ next to $u$ in a shortest path from $u$ to $v$, the algorithm may store the pair $u', v$ as a routing table entry at $u$.
Although this routing scheme gives an optimal stretch, the routing table at every node will have $O(n)$ entries.
Early work on the routing problem focused on routing schemes for the special case of general graphs such as trees \cite{conf/icalp/FragGavo01,journals/comp/SantoroK85}, planar graphs \cite{journals/jacm/Thorup04}, unit disk graphs \cite{journals/algorithmica/KaplanMRS18,journals/cgta/YanXD12}, networks of low doubling dimension \cite{journals/talg/KonjevodRX16}, and for graphs embedded in geometric spaces \cite{journals/siamjc/BoseFRV15,journals/jocg/BoseFRV17,journals/tcs/BoseMorin04}.
The routing scheme presented by Awerbuch~\etal~\cite{journals/jalgo/AwerbuchBLP90} guarantee a stretch factor of $O(k^2)$ and requires storing $O(k n^{1/k} (\lg{n}) (\lg{D}))$ bits of routing information per vertex of the input graph.
Cowen~\cite{journals/jalgo/Cowen01} designed a routing scheme that has a multiplicative stretch $3$, packet headers of size $O(\lg{n})$, and routing tables of size $\tilde{O}(n^{2/3})$. 
Thorup and Zwick~\cite{conf/spaa/ThorupZwick01,journals/jacm/ThorupZwick05} describe a routing scheme that uses $\tilde{O}(n^{1/k})$ bits of space at every node and has a stretch of $4k-5$, for every $k \geq 1$.
Chechik~\cite{conf/podc/Chechik13} devises a scheme using $\tilde{O}(n^{1/k} \lg{D})$ bits of space at every node and having a stretch $c \cdot k$ for some $c < 4$, for a sufficiently large integer $k$.
Roditty and Tov~\cite{conf/podc/RodittyTov15} designed a routing scheme using $\tilde{O}(\frac{1}{\epsilon}n^{1/k} \lg{D})$ bits of space at each node and having a stretch $4k-7+\epsilon$, for every integer $k \ge 1$.
In all these results, $D$ is the diameter of the input graph.
Peleg and Upfal~\cite{journals/jacm/PelegUpfal89} had shown that any routing scheme with a constant stretch factor needs to store $\Omega(n^c)$ bits per node for some constant $c > 0$.

\subsection*{\bf Preliminaries}
\label{subsect:prelim}

The input polygonal domain consisting of $h$ convex polygonal obstacles in the plane is denoted by $\cal P$.
The number of vertices that define the obstacles in $\cal P$ is denoted by $n$.
We assume that the obstacles in $\cal P$ are placed in a large bounding box $B({\cal P})$. 
We assume $\cal P$ is in a general position, that is, no three vertices of $\cal P$ are collinear. 
Each obstacle is associated with a unique integer in $[0, h-1]$, which is its {\it identifier}.
For any $i \in [0, h-1]$, the obstacle with the identifier $i$ is called the $i^{th}$-obstacle, and is denoted by $P_i$.
For any $i \in [0, h-1]$, the boundary of an obstacle $P_i$ is denoted by $bd(P_i)$.
The boundary $bd({\cal P})$ of $\cal P$ is the boundary of $B({\cal P})$ union the boundaries of obstacles in $\cal P$.
The {\it free space} $\cal{F(P)}$ of $\cal P$ is the closed region $B({\cal P})$ excluding the relative interiors of obstacles in $\cal P$.
For each obstacle $P_i$, starting from an arbitrary vertex of $P_i$, traversing the $bd(P_i)$ in counterclockwise direction, the vertices of $P_i$ are numbered from $0$ to $n_i-1$, where $n_i$ is the number of vertices of $P_i$.
A vertex of $P_i$ that is assigned the number $k$ is denoted by $v_{ik}$.
The identifier $id(v_{ik})$ of any vertex $v_{ik}$ is the ordered tuple consisting of the binary representation of $i$ followed by the binary representation of $k$.
Every vertex $v$ of $\cal P$ is also associated with a {\it routing label} $L(v)$, which is detailed later.
Note that for any vertex $v$ of $\cal P$, $id(v)$ is not necessarily the same as $L(v)$.
For each vertex $v \in {\cal P}$, our preprocessing algorithm computes three routing tables, named $\rho_v^1$, $\rho_v^2$, and $\rho_v^3$, and stores each of these tables with $v$.
When a packet is at a vertex $v$, and its next hop along the routing path is $w$, we say the packet is {\it forwarded} from $v$ to $w$.

To distinguish from the vertices of $\cal P$, every point in $\mathbb{R}^2$ that is not necessarily a vertex of $\cal P$ is called a point, and the vertices of graphs are called nodes.
Two points $p, q$ in $\cal{F(P)}$ are said to be {\it visible} to each other whenever the relative interior of the line segment joining $p$ and $q$ does not intersect any edge on $bd(P_i)$ for any $i$.
The length of any polygonal path in $\cal{F(P)}$ is the sum of the lengths of all line segments that belong to that path.
A geodesic path in $\cal{F(P)}$ is a simple polygonal path in $\cal{F(P)}$ such that every two successive vertices along that path are mutually visible to each other and the length of that path cannot be shortened by slight perturbations.
For any two points $p, q \in \cal{F(P)}$, among all geodesic paths between $p$ and $q$, the path that has the minimum length is a shortest (geodesic) path between $p$ and $q$. 
The shortest distance between $p$ and $q$ in $\cal{F(P)}$ is the length of a shortest path between $p$ and $q$ and is denoted by $d(p, q)$.
The Euclidean distance between any two points $p$ and $q$ in $\mathbb{R}^2$ is denoted by $\Vert pq \Vert$.

For any point $p \in \cal{F(P)}$ and any line segment $\ell \in \cal{F(P)}$, a {\it geodesic projection} $p_\ell$ of $p$ on $\ell$ is a point on $\ell$ such that $p_{\ell}$ is at the minimum geodesic distance from $p$ among all the points on $\ell$.
The {\it (geodesic) distance} from $p$ to its geodesic projection $p_\ell$ on a line $\ell \in \cal{F(P)}$ is the weight $w(p_{\ell})$ of $p_\ell$.  
For any line $\ell$ and two points $p', p'' \in \ell$ associated with non-negative weights, the {\it weighted distance} between $p'$ and $p''$ is defined to be equal to $w(p')+|p'p''|+w(p'')$.
For a vertex $v$ of $P$, let $r$ be the ray originating at $v$ oriented in the positive (resp. negative) $y$-direction such that $r$ intersects with $bd({\cal P})$.
For point $p$ being the first point of intersection of $bd({\cal P})$ with $r$ along ray $r$, line segment $vp$ is called the {\it upward (resp. downward) splitter} at $v$ in $P$.
When there is no need to qualify the splitter by upward or downward, we simply call it a splitter.
For every obstacle $P_i \in \cal P$, from the leftmost vertex $v$ of $P_i$, we introduce one upward splitter and one downward splitter. 
Also, for every obstacle $P_i \in \cal P$, from the rightmost vertex $v$ of $P_i$, we introduce one upward splitter and one downward splitter. 
All these splitters together decompose $\cal{F(P)}$ into a planar subdivision comprising the set $\cal F$ of faces with $|{\cal F}| = O(h)$.
Significantly, every face in $\cal F$ is a simple polygon with its boundary comprising at most two convex chains and at most four splitters.

Let $G'(V', E')$ be the dual graph of the planar subdivision $\cal F$.
Noting that $G'$ is a planar graph, using the planar separator theorem, we partition $V'$ into three sets $A'$, $B'$, and $S'$ such that $S'$ separates $A'$ from $B'$.
From the planar separator theorem, this results in $|A'|, |B'|= O(h)$ and $|S'| = O(\sqrt{h})$.
Since there are at most four splitters in every simple polygon corresponding to a vertex in $S'$, the total number of splitters in $S'$ is $O(\sqrt{h})$. 
And, for any two vertices $v' \in A'$ and $v'' \in B'$, any path between $v'$ and $v''$ passes through a simple polygon belonging to $S'$, and that path intersects at least one splitter in the component $S'$. 
The union of simple polygons corresponding to each node of $A'$ (resp., $B'$ and $S'$) is a path-connected component.
We call the component induced by the nodes in $A'$ (resp., $B'$ and $S'$) the {\it component $A'$} (resp., {\it component $B'$} and {\it component $S'$}), which is essentially the union of simple polygons corresponding to dual nodes in $A'$ (resp., $B'$ and $S'$).
We also say that a vertex of a simple polygon belonging to component $A'$ (resp., $B', S'$) is a vertex of $A'$ (resp., $B', S'$).

The planar separator theorem is applied recursively.
At the root of the corresponding recursion tree $T$, the planar separator theorem is applied to $V'$ to obtain $A', B'$ and $S'$.
At the child nodes of the root of $T$, the theorem is applied to each of $A', B'$, and $S'$.
Our algorithm continues to apply the planar separator theorem at every internal node of $T$.
The base case arises whenever the number of simple polygons in the corresponding component is one.
And, the simple polygon at every leaf node of $T$ belongs to $\cal F$, and no two leaf nodes of $T$ store the same simple polygon in $\cal F$.

The levels of $T$ are numbered from $0$; that is, the level of the root node of $T$ is $0$.
In every level, the nodes are numbered from left to right, starting from $0$.
The $j^{th}$ node in $i^{th}$ level of $T$ is denoted by $v_{ij}$, and the component associated with $v_{ij}$ is denoted by $P_{ij}$.
The set comprising splitters in $P_{ij}$ is denoted by $X_{ij}$.
We sort the splitters in $X_{ij}$ by their $x$-coordinates and store them in an array in this order.
A splitter that occurs in $k^{th}$ position in the sorted order of splitters in $X_{ij}$ is denoted by $\ell_{ijk}$.
Each splitter $\ell_{ijk}$ is assigned an identifier $id(\ell_{ijk})$, which is the ordered tuple consisting of binary representations of $i, j$, and $k$, in this order. 

\subsection*{\bf Our Contributions}
\label{subsect:contrib}

Our preprocessing algorithm first computes a set $\cal F$ of simple polygons by introducing upward and downward splitters at every leftmost and every rightmost vertex of every polygonal obstacle in $\cal P$.
This results in $O(h)$ simple polygons, that is, $|{\cal F}| = O(h)$.

Our algorithm recursively applies the planar separator theorem to the dual graph of polygons in $\cal F$.
At node $v_{ij}$ of the recursion tree $T$, applying the planar separator theorem results in $A_{ij}, B_{ij}$, and $S_{ij}$, each being a collection of simple polygons, such that any path between any vertex in $A_{ij}$ and any vertex in $B_{ij}$ intersects $S_{ij}$.
Hence, we call $S_{ij}$ the separator component at the node $v_{ij}$.
In case $s$ and $t$ do not lie in a simple polygon belonging to $\cal F$, then there is a node in the recurrence tree such that $s$ and $t$ are separated by a separator component.

Let $X_{ij}$ be the set comprising all the splitters in $S_{ij}$.
We compute a geodesic projection of every vertex of polygon $B_{ij}$ on every splitter in $X_{ij}$.
For every splitter $\ell_{ijk} \in X_{ij}$, as in Abam~\etal~\cite{journals/siamjc/AbamBergS19} and Bhattacharjee and Inkulu~\cite{journals/ijcga/BInkulu22}, we partition projected points on $\ell_{ijk}$ into clusters and a point from each such cluster is chosen as its representative, which is called the center of that cluster.
The vertex whose geodesic projection on $\ell_{ijk}$ resulted in that representative point is called a center vertex of $\ell_{ijk}$, or a center vertex of that cluster.
To route a packet to any vertex whose geodesic projection belongs to a cluster, that packet gets forwarded to the center vertex of that cluster.
This further optimizes routing table sizes.

To route a packet from any vertex $v$ in simple polygon $A_{ij}$ to $t \in B_{ij}$, we find the first splitter $\ell_{ijk}$ in $X_{ij}$ intersected by the shortest path from $v$ to $t$ and a specific center vertex $c$ of $\ell_{ijk}$.
To facilitate this, corresponding to every source vertex, every vertex $v$ of $\cal P$ stores both $\ell_{ijk}$ and $c$.
The concatenation of all such tuples is the routing label $L(v)$ of $v$.
Every packet that needs to be routed to any vertex $v'$ carries both $id(v')$ and $L(v')$.
In routing from $v$ to $t$, our routing algorithm routes the packet from $v$ to $c$.
The packet is routed from $v$ to $c$ along a path in the shortest path tree $T_c$ rooted at $c$.
Essentially, at every node $u$ in $T_c$, the parent node of $u$ in $T_c$ is saved in a routing table at $u$.

However, considering the space of routing tables, it is expensive to store at a vertex $v'$, a center vertex corresponding to every vertex $v''$ of $\cal P$. 
To avoid this, sections of boundaries of obstacles located in $B_{ij}$ are partitioned into intervals with respect to each source vertex $v \in A_{ij}$. 
An interval consists of a contiguous sequence of vertices along the boundary of an obstacle wherein that section of boundary is located in $B_{ij}$.
Each such interval is associated with the identifier of splitter $\ell$ to which this interval corresponds.
We show that there are only $O(h)$ intervals corresponding to each source vertex $v$.
Further, our research identifies the intersection points of shortest paths from the vertex $v$ to the vertices of any interval associated with $v$ and splitter $\ell$.
The interval formed by these points of intersection on $\ell$ is called a band of that interval. 
These bands further help efficiently compute routing tables by establishing a correspondence between intervals and bands.

At every vertex $v$, we maintain three routing tables: $\rho_v^1, \rho_v^2$, and $\rho_v^3$.
Using $id(t)$, we search in $\rho_v^1$ to determine the interval to which $t$ belongs.
The entry in that interval contains the identifier of a splitter $\ell'$.
Using $\ell'$, we find the center vertex $c'$ in $L(t)$.
Once the center vertex $c'$ is determined, we use the routing table $\rho_v^2$ to find the next hop on the routing path to $c'$.
Upon the packet reaching $c'$, we recursively forward the packet to $t$, unless $c'$ and $t$ belong to the same simple polygon $P'$ belonging to $\cal F$.
In this case, the routing algorithm takes the help of $\rho_{c'}^3$.
The $\rho_{c'}^3$ has the information to forward the packet to the next hop, which is located in $P'$.
The packet will be forwarded via intermediate hops in $P'$, using next hop information stored in $\rho_u^3$ for every intermediate hop $u$, until it reaches $t$.
The following theorem summarizes our result.

\begin{theorem}
Given a convex polygonal domain $\cal P$ and an input parameter $\epsilon > 0$, the preprocessing algorithm assigns a unique label of size $O(\sqrt{h} (\lg{h}) \lg{n})$ to each vertex of $\cal P$ and it computes routing tables at every vertex of $\cal P$ of size $O(h\lg{n} + \sqrt{h}(\lg{h})(\min((\frac{1}{\epsilon})^{O( \lg {\alpha})},n))\lg {n})$ in $O(n^2(\lg{n}))$ time so that any packet is routed along a geodesic path with $(7 + \epsilon)(\lg{h})$ multiplicative stretch, while that packet header carries at most $2 \lg{n}$ bits of routing information. 
Here, $h$ is the number of convex polygonal obstacles in $\cal P$, $n$ is the number of vertices of $\cal P$, and $\alpha > 1$ is a geometric parameter.
\end{theorem}

When the polygonal domain consists of simple polygonal obstacles, Banyassady~\etal~\cite{journals/cgta/BanyaCKMR20} show that the routing table space is $O((h + \frac{1}{\epsilon})\lg {n}))$ bits per vertex.
When the polygonal domain consists of convex polygonal obstacles, our algorithm shows that the routing table space is $O(h\lg{n} + \sqrt{h}(\lg{h})$ $(\min((\frac{1}{\epsilon})^{O( \lg {\alpha})},n))\lg {n})$ bits per vertex.
If we consider the case where $\left(\frac{1}{\epsilon}\right)^{O(\lg \alpha)}$ is smaller than $n$, $\alpha > 1$, and $\epsilon$ belongs to $(1, \sqrt{h})$, the size of routing tables of our algorithm are of better size.
Further, the preprocessing time of algorithm in Banyassady~\etal~\cite{journals/cgta/BanyaCKMR20} is $O(n^2\lg{n} + \frac{n}{\epsilon})$.
For a convex polygonal domain, our preprocessing algorithm runs in $O(n^2 \lg{n})$ time. 
This algorithm's preprocessing time improves on the former when $\frac{1}{\epsilon}$ is $\omega(n\lg{n})$.
When the polygonal domain consists of convex polygonal obstacles, Inkulu and Kumar~\cite{journals/ijfcs/InkuluKum24} show that the total routing table space is $O(n + \frac {h^{3}}{\epsilon^{2}}polylog(\frac{h}{\epsilon}))$ and the preprocessing time is $O(n + \frac {h^{3}}{\epsilon^{2}}polylog(\frac{h}{\epsilon}))$, for $\epsilon <1$. 
Their result holds only for $\epsilon < 1$, whereas the algorithm in this paper holds for every $\epsilon > 0$.
To our knowledge, this is the first routing scheme for polygonal domains that employs divide-and-conquer for preprocessing, subdivides the polygonal domain with splitters, and applies the planar separator theorem to the resultant planar subdivision.
In addition, for each component resulting from applying the planar separator theorem, our algorithm selects a set of splitters located within that component and geodesically projects a subset of vertices that define the component onto those splitters.
The resulting points of projections on each such splitter are partitioned into clusters, and a special vertex from the component is chosen for each such cluster.
The routing scheme is designed such that a subset of vertices in these components will always forward packets to this special vertex.
Upon a packet reaching a special vertex, it is further routed recursively until it reaches its destination.
So both our preprocessing algorithm and the routing phase are recursive.
As a byproduct, for routing in the base case of recursion, we devise an efficient routing algorithm for a special simple polygon, bounded by two convex chains and two polygonal chains on a vertical line segment. 
These optimizations, together with the organization of routing tables, helped achieve an efficient preprocessing time and improved routing table space bounds.

Section~\ref{sect:preproc} details the preprocessing algorithm.
This involves recursively computing planar subdivisions, splitters in each subdivision, the center vertices corresponding to those splitters, and identifying intervals and bands.
The routing phase is detailed in Section~\ref{sect:routing}.
Conclusions are in Section~\ref{sect:conclu}.

\section{Preprocessing Algorithm}
\label{sect:preproc}

Given a convex polygonal domain $\cal P$, for every $v$ of $\cal P$, the preprocessing algorithm computes the label $L(v)$ of $v$ and three routing tables, named $\rho_v^1, \rho_v^2$, and $\rho_v^3$, and stores them at $v$.
For a packet originated at any vertex $v$ of $\cal P$ with destination vertex being $t$ of $\cal P$, these routing tables help in routing the packet to reach $t$, as the packet hops via intermediate vertices of $\cal P$.
At every intermediate vertex $v'$ along this routing path, the next hop is determined using routing tables at $v'$ together with the information in the packet header. 

In general, in the routing phase, at $s$, by traversing $\rho_s^1$, a {\it center vertex} $c$ is determined based on $L(t)$ stored in the packet header. 
The $id(c)$ is set as the {\it pseudo-destination} of the packet and is stored in the packet header.
This says the routing scheme's immediate objective is to route the packet to $c$, and if $c \ne t$, the packet will thereafter be routed to $t$.
From there on, until the packet reaches $c$, at every intermediate hop $v$, the next hop is determined from the routing table $\rho_v^2$. 
Once the packet reaches $c$, the next center vertex $c'$ is determined from $\rho_{c}^1$ and $L(t)$, and $id(c')$ is set as the pseudo-destination of that packet.
Essentially, at any vertex $v$ of $\cal P$, if the packet's pseudo-destination is $id(v)$, then by using $L(t)$, from the routing table $\rho_v^1$, the next pseudo-destination for the packet is determined.
Otherwise, the routing table $\rho_v^2$ is traversed to determine the next hop for the packet.
Intuitively, as the packet moves from one center vertex to the next center vertex, it gets closer to its destination.
As described below, our algorithm partitions $\cal{F(P)}$ into a set $\cal F$ of simple polygons.
The above-mentioned routing scheme is applied at every hop along the routing path, until the packet reaches a center vertex $c''$ such that both $c''$ and $t$ are located in the same polygon $f \in {\cal F}$.
Then, $\rho_{c''}^3$ is traversed to determine the next hop located in $f$.

The algorithm first assigns a unique identifier in $[0, h-1]$ to each convex polygonal obstacle in $\cal P$.
For each hole $P_i$ of $\cal P$, each vertex of $bd(P_i)$ is assigned a unique number: $0$ is assigned to any arbitrarily chosen vertex $v'$ of $P_i$, and while traversing $bd(P_i)$ in counterclockwise direction from $v'$, any new vertex encountered is associated with a number equal to the number last used for a vertex of $P_i$ incremented by one.
And, the identifier of any vertex of $P_i$ is the ordered tuple consisting of the binary representation of $i$ followed by the binary representation of this number. 

Our algorithm computes the planar subdivision $\cal F$ induced by upward and downward splitters incident to the leftmost and rightmost vertices of each obstacle in $\cal P$. 
We compute $\cal F$ by a plane sweep of $\cal P$ in $B({\cal P})$.
For every obstacle $P_i \in {\cal P}$, the leftmost and rightmost vertices of $P_i$ are marked.
These marked points are sorted by their $x$-coordinates.
Whenever the sweep line encounters any of these marked vertices, the algorithm introduces upward and downward splitters into $\cal{F(P)}$.
Specifically, when $\ell$ is at a marked vertex $v'$, the nearest point of intersection of $bd({\cal P})$ with $\ell$ above $v'$ is joined with $v'$.
Analogously, the nearest point of intersection of $bd({\cal P})$ with $\ell$ below $v'$ is also joined with $v'$.
These two line segments are the upward and downward splitters at $v'$.
This plane sweep is implemented using two data structures, a binary search tree $\tau$ and an event queue $Q$.
At any juncture of plane sweep, the former stores the sorted $y$-order in which $\ell$ intersects obstacles in $bd({\cal P})$, and the event queue $Q$ stores marked vertices that are yet to be encountered by $\ell$ in sorted order of their respective $x$-coordinates.
Initially, all the marked vertices are inserted into $Q$.
These data structures are updated whenever $\ell$ is at an event point that is popped from $Q$.
We utilize the convexity of obstacles by finding the points of intersection of $\ell$ with any obstacle $P_i$ in $O(\lg{h})$ time.
Hence, it suffices for the sweep line to update data structures and compute splitters by stopping only at the marked vertices, rather than updating them at every vertex.
At a marked vertex $v$ which is the leftmost (resp. rightmost) vertex of any obstacle $P_i$, we remove (resp. include) $P_i$ from (resp. into) $\tau$.
And, when an event occurs, that event point is removed from $Q$.
The resulting subdivision $\cal F$ is stored in a doubly connected edge list. 
(See de Berg~\etal~\cite{books/compgeom/deberg2008}.)
It is immediate that this plane sweep takes $O(n+h\lg{h})$ time in the worst case.

For every vertex $v \in \cal P$, using Hershberger and Suri's algorithm \cite{journals/siamcomp/HershbergerS99}, we construct a shortest path tree $T_v$ rooted at $v$.
Since the shortest path between any two vertices of a polygonal domain is not necessarily unique, our algorithm considers one such shortest path and builds routing tables based on that.
We assume the shortest path from $v$ to any vertex $w$ in $\cal P$ is the path from $v$ to $w$ in $T_v$.
These trees are computed only once and reused as $\cal{F(P)}$ gets recursively decomposed.
Given a packet with its pseudo-destination field set to a vertex $c$, during preprocessing, for any center vertex $c$ of a splitter $\ell_{ijk}$, the preprocessing algorithm uses the shortest path tree $T_c$ rooted at $c$ for storing the parent of $v$ in $T_c$ with $v$ for every $v$ in $T_c$.
As part of routing the packet from $v$ to $c$, if the packet arrives at $v$, this helps $v$ route it to its parent in $T_c$.
In addition, as detailed below, these shortest-path trees help find the first splitter in the separator component that is intersected by a shortest path between two vertices.

Without loss of generality, we assume from here on that $v$ is located in $A_{ij}$ and $t$ is located in $B_{ij}$.
Our routing algorithm and analysis also extend to other combinations of source and destination locations, as described herewith.
When routing a packet from a vertex $v \in A_{ij}$ to a vertex $t \in S_{ij}$, the center vertex that occurs after $v$ on the routing path from $v$ to $t$ corresponds to a splitter in $X_{ij}$, and that center vertex belongs to $S_{ij}$.
And, when routing a packet from a vertex in $v \in S_{ij}$ to a vertex in $t \in B_{ij}$, the center vertex that occurs after $v$ on the routing path from $v$ to $t$ also corresponds to a splitter in $X_{ij}$, and that center vertex belongs to $B_{ij}$.
Routing scheme within each region comprising $A_{ij}$, $B_{ij}$, and $S_{ij}$ is constructed recursively.
We note that to facilitate routing from a vertex in $B_{ij}$ to a vertex in $S_{ij}$, from a vertex in $B_{ij}$ to $A_{ij}$, or a vertex in $S_{ij}$ to a vertex in $A_{ij}$, the same algorithm devised herewith is applied by reversing the roles of $A_{ij}$ and $B_{ij}$.
The subsections below detail how our preprocessing algorithm computes $L(v)$ and routing tables $\rho_v^1, \rho_v^2$, and $\rho_v^3$ at every vertex $v$ of $\cal P$.

\subsection{\bf Recursive Subdivision}
\label{subsect:recdecomp}

As described, our algorithm first partitions $\cal{F(P)}$ into simple polygons by introducing upward and downward splitters.
By introducing a node for each of these simple polygons and an edge between every two nodes whenever the simple polygons corresponding to those two nodes abut along any splitter, we construct the dual graph $G$.
Further, we associate $\cal{F(P)}$ with the root node of the recursion tree. 
Next, our algorithm partitions $G$ into three sets by applying the following planar separator theorem by Lipton and Tarjan~\cite{journals/siamam/LiptTarj79}.

\begin{theorem}\cite{journals/siamam/LiptTarj79}
\label{thm:plsep}
For a planar graph $H(V, E)$ with $|V| = m$, the vertex set $V$ of $H$ can be partitioned into three sets $A, B$, and $S$, such that
(i) $|S| = O(\sqrt{m})$, 
(ii) $|A|, |B| \le \frac{2}{3}m$, and
(iii) no edge in $E$ is incident to a vertex in $A$ and a vertex in $B$.
\end{theorem}

Significantly, by applying this theorem to $G$, we partition nodes in $G$ into three sets $A_{00}, S_{00}$, and $B_{00}$ such that for any two nodes $v' \in A_{00}, v'' \in B_{00}$, every path in $G$ between $v'$ and $v''$ passes through a vertex in $S_{00}$.
Hence, $S_{00}$ is called a separator of $G$.
Refer to Fig.~\ref{fig:plsubdiv}.
Since each node of $G$ corresponds to a simple polygon, with a slight abuse of notation, we say $A_{00}$ (resp., $B_{00}, S_{00}$) is a collection of simple polygons.
In the recursion tree, $A_{00}, S_{00}, B_{00}$ are respectively associated with three child nodes of the root node.
We recursively apply planar separator theorem to each of $A_{00}, S_{00}$, and $B_{00}$.
From here on, we explain the algorithm with respect to an internal node $v_{ij}$ of the recursion tree that is associated with component $P_{ij}$.
That is, by applying the planar separator theorem to the dual graph $G'$ of simple polygons in $P_{ij}$, we partition the vertices of $G'$ into three sets, denoted by $A_{ij}, S_{ij}$ and $B_{ij}$, such that for any node $v'$ in $A_{ij}$ to any node $v''$ in $B_{ij}$, any path from $v'$ to $v''$ has to pass through a simple polygon in $S_{ij}$.

\begin{figure}[ht]
\centering
\includegraphics[width=6.5cm]{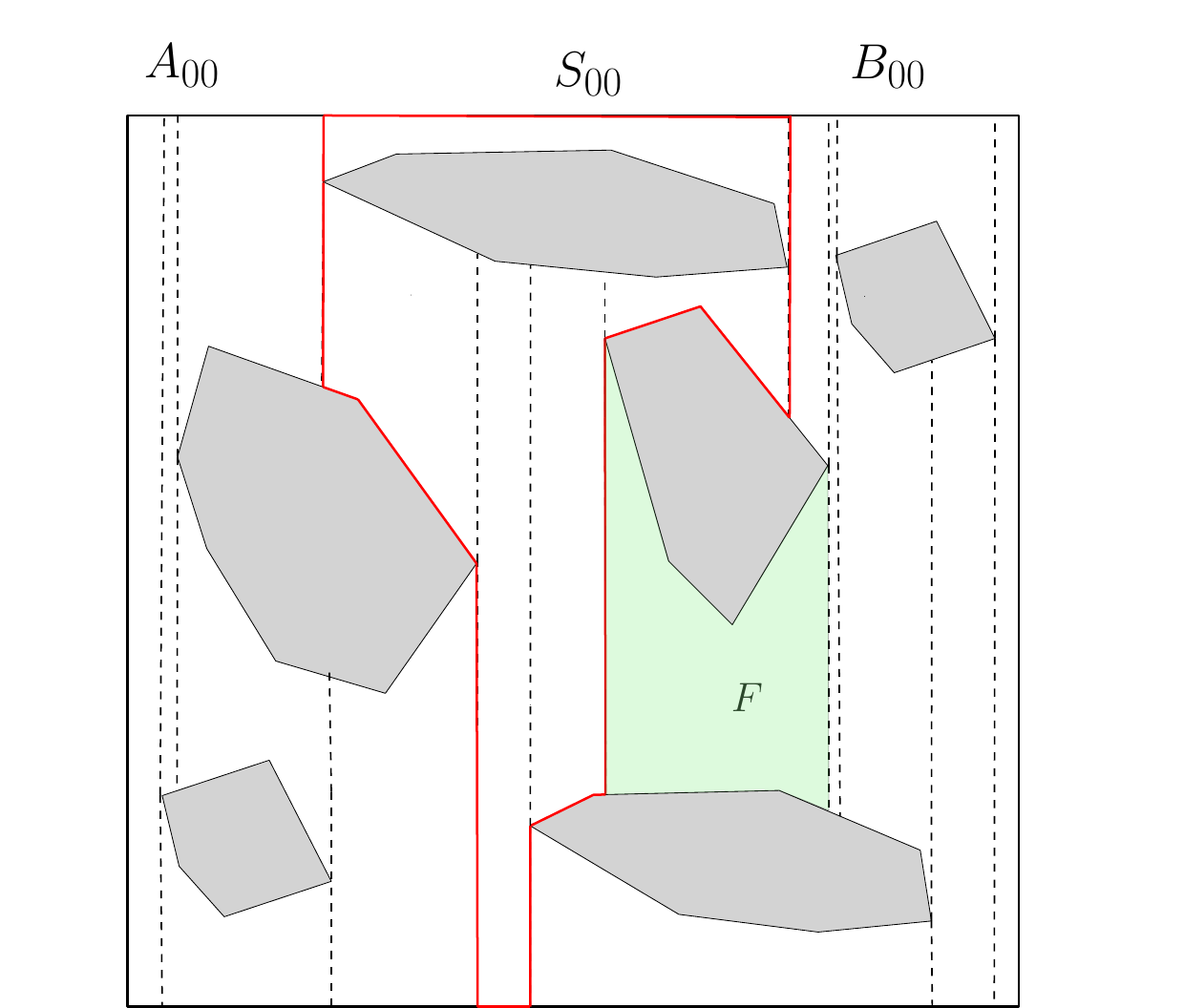}
\caption{
\footnotesize
Illustrating the partitioning of $\cal{F(P)}$ into a set $\cal F$ of simple polygons by introducing upward and downward splitters at each of the leftmost and rightmost vertices of every obstacle in $\cal P$.
With the polygon in green, the boundary of any simple polygon in $\cal F$ has at most two convex chains and at most four splitters, as illustrated.
Three components $A_{00}, S_{00}$, and $B_{00}$ resulting from applying the planar separator theorem to dual nodes of simple polygons in $\cal F$ are also shown.
The boundary of the union of simple polygons in $S_{00}$ is shown in red.
\normalsize
}
\label{fig:plsubdiv}
\end{figure}
The following lemma shows that any shortest path from any vertex $v$ of $A_{ij}$ to any vertex $t$ of $B_{ij}$ intersects at least one splitter in $X_{ij}$.
To remind, $X_{ij}$ is the set of splitters in $S_{ij}$.
\begin{lemma}
For a vertex $s \in A_{ij}$ and $t \in B_{ij}$, there exists at least one splitter $\ell_{ijk} \in X_{ij}$, such that the shortest path from $s$ to $t$ intersects $\ell_{ijk}$.
\end{lemma}
\begin{proof}
Consider the shortest path $\pi(s,t)$ from $s \in A_{ij}$ to $t \in B_{ij}$.
Due to the planar separator theorem, the components $A_{ij}$ and $B_{ij}$ are separated by component $S_{ij}$.
That is, there is a simple polygon $P \in S_{ij}$ such that $\pi(s, t)$ intersects $P$.
We know each of the polygons in $S_{ij}$ is a polygon in $\cal F$.
And, for every such $P' \in \cal F$, there are at most two splitters whose $x$-coordinates are the same as the $x$-coordinate of the leftmost vertex of $P'$.
For $P$, without loss of generality, let $\ell', \ell''$ be these two splitters.
Specifically, a path from $s$ to $t$ intersects $P$ only if that path intersects one of $\ell'$ and $\ell''$.
\end{proof}

Using the planar separator theorem, the following lemma gives a worst-case upper bound for $\sum_{i, j} |X_{ij}|$.
Here, each $i, j$ tuple denotes an internal node of the recursion tree.

\begin{lemma}
\label{lem:numsplitt}
The sum of the number of splitters in all the separator components at all the internal nodes of the recursion tree is $O(\sqrt{h}\lg{h})$.
\end{lemma}
\begin{proof}
Let $T(h)$ be the sum of the number of splitters in all the separators at all the internal nodes of the recursion tree $T$ when the root node of $T$ has $h$ splitters.
Then, due to the planar separator theorem applied in the above algorithm,
$T(h) \le T(2h/3) + T(k \sqrt{h}) + T(h/3 - k \sqrt{h}) + k \sqrt{h}$, for an integer $k \ge 1$ when $h$ is $\omega(1)$, and $T(h)$ is $O(1)$ when $h$ is $O(1)$.
Solving this recurrence with the guess and substitute method yields $T(h)$ is $O(\sqrt{h}\lg{h})$.
\end{proof}
 
\subsection{\bf Computing Centers}
\label{subsect:compcent}

In this section, we describe an algorithm to compute a set of center vertices corresponding to each vertex $w \in B_{ij}$ and each splitter $\ell \in X_{ij}$ so that any packet from any vertex $u \in A_{ij} \cup S_{ij}$ that needs to be routed to $w$ gets forwarded to one of these center vertices.
Among these center vertices, the specific center vertex that forwards packets from $u$ is chosen, based on the first splitter in $X_{ij}$, through which a shortest path from $u$ to $w$ intersects.
The geodesic projection of every vertex of $\cal P$ in $B_{ij}$ on every splitter in $X_{ij}$ is computed.
For each such point of projection $v_\ell$ on splitter $\ell \in X_{ij}$, $v_\ell$ is associated with a weight equal to the geodesic distance between $v_\ell$ and the vertex of $B_{ij}$ whose projection is $v_\ell$. 
For each splitter $\ell$ in $X_{ij}$, based on these weights, a clustering $\cal C$ of points of projections on $\ell$ is computed.
For each cluster in each clustering of points, a center is defined.
Among all the splitters in $X_{ij}$, let $\ell'$ be the first splitter a shortest path from $u$ to $w$ intersects.
Let ${\cal C}'$ be the clustering of points of projections on $\ell'$ wherein each such point is resulting from the geodesic projection of a vertex in $B_{ij}$.
Also, let $C' \in {\cal C}'$ be the cluster to which the geodesic projection of $w$ belongs.
Then, for the center $c'$ of cluster $C'$, packets from $u$ are forwarded to the center vertex $v'$ whose geodesic projection is $c'$.
Once the packet arrives at $v'$, if its destination $w$ is not $v'$, it is recursively forwarded from $v'$ to $w$.
Analogously, to facilitate routing from any vertex in $A_{ij}$ to any vertex in $S_{ij}$, we geodesically project vertices of $S_{ij}$ on $X_{ij}$, and compute cluster centers and the corresponding center vertices for these points of projection separately.
Below, we detail all these steps.

\begin{figure}[ht]
\centering
\includegraphics[width=6cm]{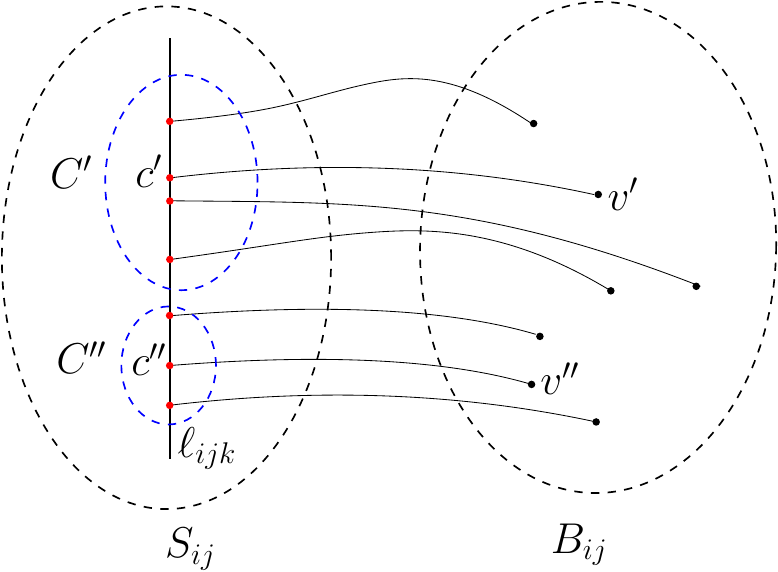} 
\caption{ 
\footnotesize
Illustrating the geodesic projection of vertices in $B_{ij}$ onto splitter $ \ell_{ijk}$, and the partitioning of projected points into two clusters.
(in blue).
Also, a center $c'$ for cluster $C'$ (resp. center $c''$ for cluster $C''$) and cluster vertex $v'$ corresponding to $c'$ (resp. cluster vertex $v''$ corresponding to $c''$) are shown.
\normalsize
}
\label{fig:centervert}
\end{figure} 

Consider any splitter $\ell_{ijk}$ in $X_{ij}$.
Let $R'$ be the set comprising points of projections of the vertices of $B_{ij}$ on $\ell_{ijk}$.
As mentioned, every point $v' \in R'$ is associated with the vertex $v$ whose geodesic projection on $\ell_{ijk}$ resulted in $v'$ and the length of the geodesic distance between $v$ and $v'$.
Our algorithm partitions points in $R'$ into a set $\cal C$ of clusters using the algorithm in Abam~\etal~\cite{journals/siamjc/AbamBergS19} (which was also in Bhattacharjee and Inkulu~\cite{journals/ijcga/BInkulu22}).
This clustering algorithm sorts the points in $R'$ in non-decreasing order of non-negative weights associated with each of these points, and considers these points in this sorted order.
As the algorithm progresses, more points are added to the current set of clusters, and additional clusters may also be initiated.
Among all the points in any cluster $C$, the first point $p$ that gets into $C$ while applying this algorithm is the representative of $C$, called the {\it center} of cluster $C$ or a center of the splitter $\ell_{ijk}$.
The vertex in $B_{ij}$ whose geodesic projection is $p$ is called a {\it center vertex} of $\ell_{ijk}$ or the center vertex of $C$.
Refer to Fig.~\ref{fig:centervert}.
For each point $p$, the cluster center $c'$ nearest to $p$ is determined from the current set of cluster centers.
If the Euclidean distance between $c'$ and $p$ is at most $\epsilon \cdot w(p)$, then we include $p$ into a cluster to which $c'$ belongs; otherwise, a new cluster with $p$ as its center is initiated.
The positive real number $\epsilon$ is an input parameter, and this particular check helps in upper-bounding the stretch factor of the routing path.
This algorithm partitions $R'$ into clusters and computes a center for each such resultant cluster.

To remind, for any splitter $\ell$ and two points $p', p'' \in \ell$, the weighted distance between $p'$ and $p''$ is defined to be equal to $w(p')+|p'p''|+w(p'')$.
Let $C_{ijk}'$ be the set of cluster centers defined due to points projected on splitter $\ell_{ijk}$.
The following analysis upper bounds the $|C_{ijk}'|$ and the time complexity of this clustering algorithm.

\begin{lemma}
\label{lem:numcentpersplitt}
The cardinality of $C_{ijk}'$ is upper bounded by $(\frac{1}{\epsilon})^{O(\lg {\alpha}_{ijk})}$, where $\alpha_{ijk}$ is the ratio of the maximum weighted distance between any two points in $C_{ijk}'$ to the minimum weighted distance between any two points in $C_{ijk}'$ and $\epsilon > 0$ is an input parameter.
\end{lemma}
\begin{proof}
From Lemma~2.1 in Abam and de Berg~\cite{journals/siamjc/AbamBergS19}, the doubling dimension of the metric space induced by points in $C_{ijk}'$ with weighted distance metric is $O(\lg(\frac{1}{\epsilon}))$.
Since the weighted distance between any two points in $C_{ijk}'$ is lower bounded by $d_{min}$, as observed in Gottlieb and Roditty~\cite{conf/esa/GottliebR08}, the number of points in $C_{ijk}'$ that are within weighted distance $x$ of any $p \in C_{ijk}'$ is $(\frac{x}{d_{min}})^{O(\lg(\frac{1}{\epsilon}))}$.
Since $x$ is upper-bounded by $\alpha_{ijk} d_{min}$, the number of points in $C_{ijk}'$ is $\alpha_{ijk}^{O(\lg {\frac{1}{\epsilon}})}$.
\end{proof}

\begin{lemma}
\label{lem:numcenters}
The total number of center vertices computed in the entire algorithm, with all possible splitters considered, is $O(\sqrt{h}(\lg{h})(\min((\frac{1}{\epsilon})^{O( \lg {\alpha})},n))$.
Here, $\alpha$ is the ratio of the maximum weighted distance between any two points to the minimum weighted distance between any two points among points projected to any splitter considered by the algorithm, and $\epsilon > 0$ is a real number.
\end{lemma}
\begin{proof}
From Lemma~\ref{lem:numsplitt}, the total number of splitters across all separator components in all internal nodes of the recursion tree is $O(\sqrt{h} \lg{h})$. 
By Lemma~\ref{lem:numcentpersplitt}, and by noting that $\alpha$ is the maximum among all $\alpha_{ijk}$s, the number of center vertices corresponding to any one splitter is bounded by $O((\frac{1}{\epsilon})^{O( \lg {\alpha})})$. 
Thus, aggregating over all splitters, the total number of center vertices is upper-bounded by $O((\frac{1}{\epsilon})^{O( \lg {\alpha})} (\sqrt{h} \lg{h}))$.
However, since the number of geodesic projections on any splitter is at most $n$, the number of center vertices per splitter is also upper bounded by $n$. 
Therefore, the number of centers per splitter is at most $\min((\frac{1}{\epsilon})^{O( \lg {\alpha})}, n)$. 
Multiplying this by the total number of splitters gives the desired bound, $O((\sqrt{h} \lg{h}) \min((\frac{1}{\epsilon})^{O( \lg {\alpha})}, n))$.
\end{proof}

\begin{lemma} 
\label{lem:timcompcen}
Computing all the center vertices corresponding to all the splitters considered by the preprocessing algorithm together takes $O(n^2 \lg{n})$ time.
\end{lemma}
\begin{proof}
For any splitter $\ell_{ijk}$, every vertex in $S_{ij} \cup B_{ij}$ is projected onto $\ell_{ijk}$, resulting in a set $R'$ of $O(n)$ points. 
Sorting these points, as part of the clustering algorithm, takes $O(n \lg n)$ time.
By maintaining the set consisting of centers of clusters in a balanced binary search tree, the nearest center to any point in $R'$ can be determined in $O(\lg {n})$ time. Thus, computing all the centers for one splitter takes $O(n \lg {n})$ time.
From Lemma~\ref{lem:numsplitt}, the total number of splitters across all levels of the recursion tree is $O(\sqrt{h} \lg{h})$. 
Therefore, the total time to compute the center vertices over all splitters is $O(n (\lg {n})(\sqrt{h} \lg{h}))$.
Since $\sqrt{h} \lg{h} < n$, the overall time complexity is as stated.
\end{proof}

\subsection{\bf Computing Intervals}

For every vertex $v$ of $\cal P$, to help in routing from $v$ to every possible destination vertex $t$ in the routing phase, if the preprocessing algorithm stores the center corresponding to $t$ at $v$, then the overall space required would be $\omega(n^2)$.
Instead, for every splitter $\ell$ on which $t$ gets projected, with $C$ being the cluster of $\ell$ to which $t$'s geodesic projection belongs, the tuple comprising the identifier of $\ell$ together with the identifier of the center vertex of $C$ is stored in the routing label $L(t)$ of $t$.
We note that, apart from $id(t)$, the packet carries $L(t)$ in its header.
Further, we partition the vertices in $B_{ij}$ into intervals.
To remind, an interval is a contiguous sequence of vertices along the boundary of an obstacle, and this partitioning is with respect to a source vertex located in $A_{ij}$.
In the routing phase, at any vertex $v$, using the $id(t)$ stored in the packet, the interval $I$ of vertices to which $t$ belongs is found from the routing table $\rho_v^1$.
In that entry of the routing table, corresponding to $I$, an identifier of splitter $\ell$ is saved.
From $L(t)$ stored in the packet header, the center vertex $c$ that corresponds to $\ell$ is extracted.
And then the packet is forwarded to pseudo-destination $c$.
As shown later, this organization of routing information into $L(t)$ significantly reduces the total size of routing tables at the expense of storing a routing label in the packet header.
Below, we detail the characterizations required to identify intervals and to devise an efficient algorithm to compute them.

\begin{figure}[ht]
\centering
\includegraphics[width=8cm]{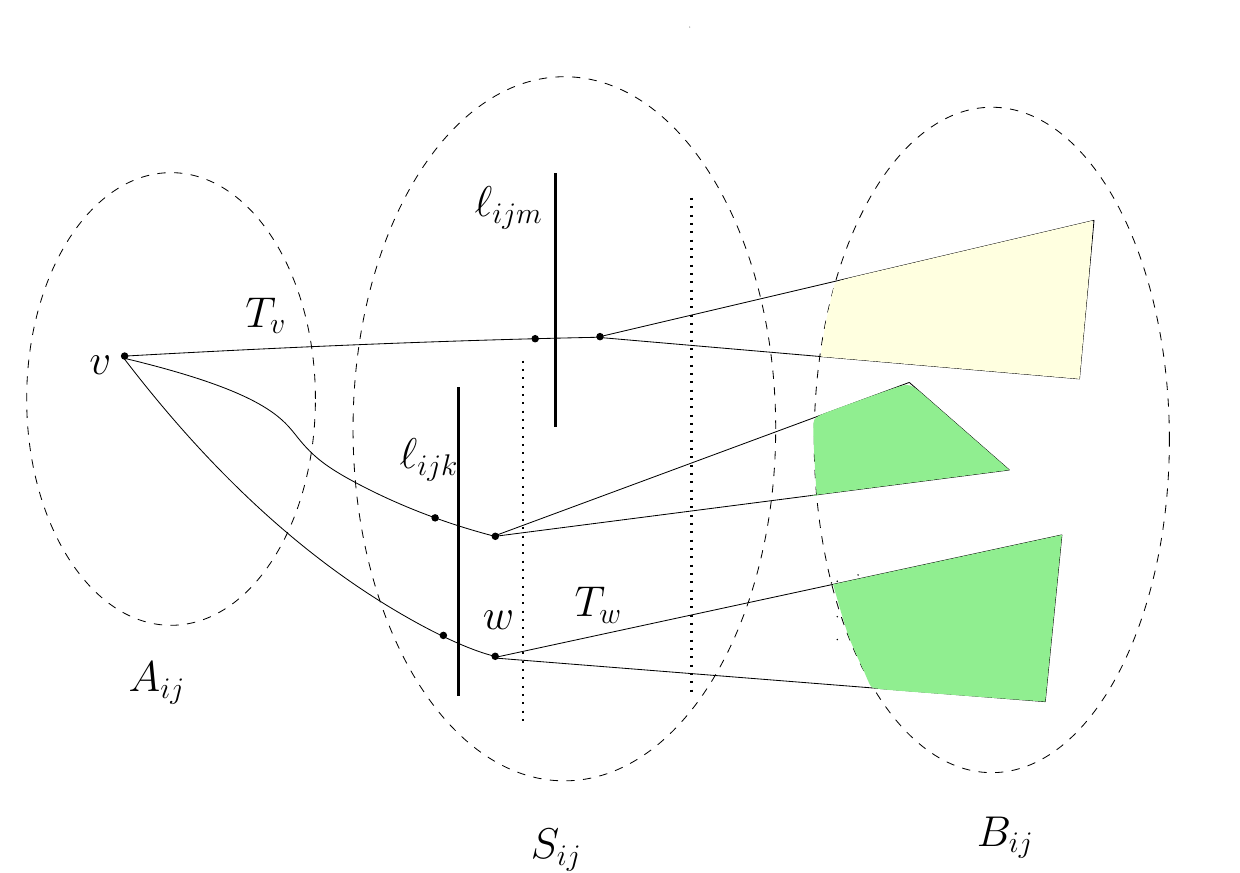}

\caption{\footnotesize
Illustration of the vertex sets $V_{ijk}$ and $V_{ijm}$ for $0 \leq k, m < |X_{ij}|$.
The vertices in set $V_{ijk}$ (resp. $V_{ijm}$) are shown in green (resp. yellow).
For the first vertex $w$ to the right of the first splitter $\ell_{ijk}$ in $X_{ij}$ intersected by a path in $T_v$, all the vertices in $T_w$ are included in $V_{ijk}$.
Analogously, other vertices colored green (resp. yellow) are also included in $V_{ijk}$ (resp. $V_{ijm}$).
\normalsize
}
\label{fig:vijk}
\end{figure}

Let $\ell_{ijk}$ be the first splitter in $X_{ij}$ that is intersected by the shortest path from $v \in A_{ijk}$ to $t \in B_{ij}$.
Also, let $c$ be the center vertex of cluster $C$ wherein $t$'s geodesic projection on $\ell_{ijk}$ belongs to $C$.
Then, while routing from $v$ to $t$, we first route the packet from $v$ to $c$.
Let $V_{ijk}$ be the set comprising every vertex $w$ in $B_{ij}$ such that the shortest path from $v$ to $w$ first intersects $\ell_{ijk}$ among all the splitters in set $X_{ij}$. 
Refer to Fig.~\ref{fig:vijk}.
Also, let $V_{ijk}^r \subseteq V_{ijk}$ be the subset of vertices in $V_{ijk}$ that lie on obstacle $P_r$.
Let $v_s$ (resp. $v_e$) be the first vertex (resp. last vertex) in $V_{ijk}^r$ that occurs in traversing the $bd(P_r)$ in counterclockwise direction.
The vertices in $V_{ijk}^r$ along $bd(P_r)$ define an {\it interval}, denoted by $(id(v_s), id(v_e), id(\ell_{ijk}))$.
That is, in traversing the $bd(P_r)$ in counterclockwise direction from $v_s$ to $v_e$, every vertex that occurs belongs to this interval.
Due to the non-crossing property of shortest paths, the vertices in $V_{ijk}^r$ are contiguous on the boundary of $P_r$.
That is, for no vertex $w$ in this interval, shortest path from $v$ to $w$ first intersects splitter $\ell_{ijk}'$ ($\ne \ell_{ijk}$) in $X_{ij}$.
Hence, any two intervals induced by any source vertex on any obstacle are pairwise disjoint.
We again note that each of these intervals is with respect to a vertex $v$ in $A_{ij}$, and the shortest-path tree used to define these intervals is rooted at $v$. 
Each such interval due to source $v$ is stored in the routing table $\rho_v^1$. 
\begin{figure}[ht]
\centering
\includegraphics[width=8cm]{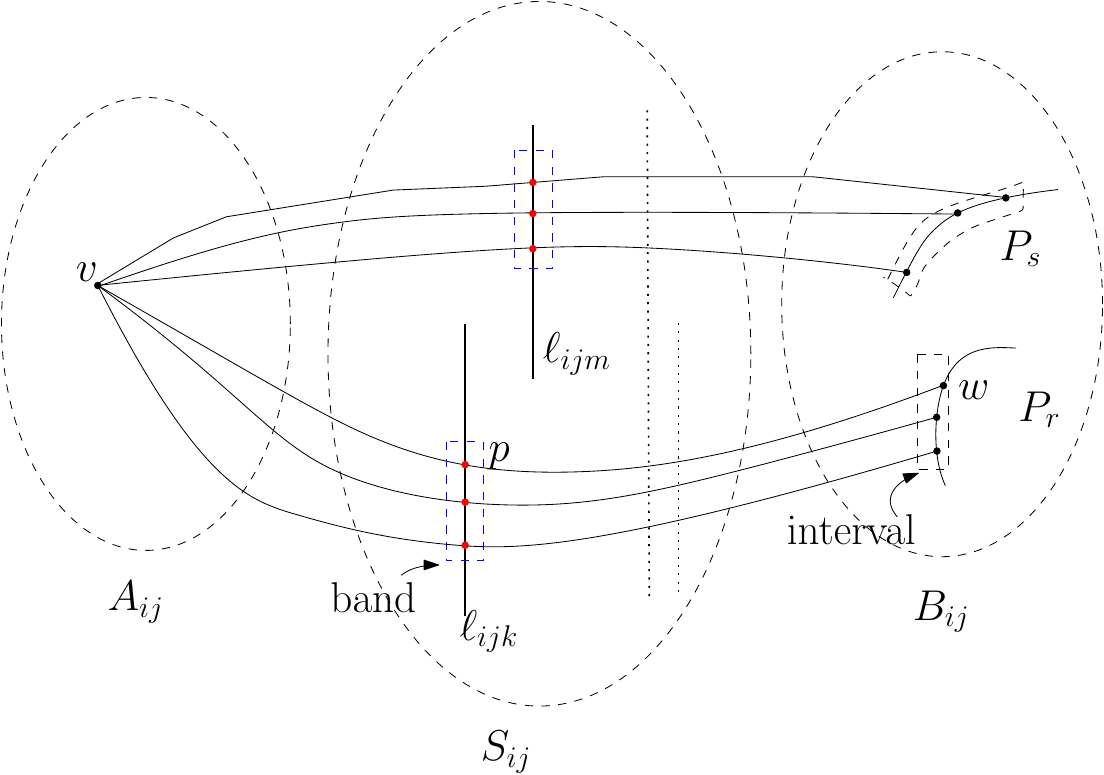}
\caption{
\footnotesize
An interval on obstacle $P_r$ (resp. $P_s$) due to source vertex $v$ and the band corresponding to it on $\ell_{ijk}$ (resp. $\ell_{ijm}$) are illustrated.
\normalsize
}
\label{fig:intervals}
\end{figure}

For an interval $I(id(v_s), id(v_e), id(\ell_{ijk}))$, let $p_s$ (resp. $p_e$) be the point at which the shortest path from $v$ to $v_s$ (resp. $v_e$) intersects $\ell_{ijk}$.
For every vertex $w$ that belongs to $I$, due to the non-crossing property of shortest paths, the shortest path from $v$ to $w$ intersects $\ell_{ijk}$ at a point that belongs to the line segment $p_sp_e$ on $\ell_{ijk}$. 
This is ensured again due to the non-crossing property of shortest paths.
A {\it band} is the maximal line segment on any splitter $\ell_{ijk}$ consisting of points of intersections of shortest paths from source vertex $v$ to vertices belonging to any interval that corresponds to $v$ and $\ell_{ijk}$.
Refer to Fig.~\ref{fig:intervals}.
Hence, any band always corresponds to a source vertex and a splitter.
Since intervals induced by any source vertex $v$ and splitter $\ell_{ijk}$ combination are pairwise disjoint, due to the non-crossing property of shortest paths, the bands formed by these intervals on $\ell_{ijk}$ are also pairwise disjoint. 

Since no two vertices of $\cal P$ have the same $x$-coordinate, splitters can be ordered with respect to their $x$-coordinates.
For any shortest path $SP_{vw}$ from $v$ to any vertex $w \in B_{ij}$, among all the points of intersections of $SP_{vw}$ with splitters in $X_{ij}$, the point of intersection that has the least $x$-coordinate is uniquely defined.
This is the first point of intersection of $SP_{vw}$ with any of the splitters in $X_{ij}$.
Hence, any two bands, each on a distinct splitter, define two distinct intervals.
Indeed, all bands on all splitters in $X_{ij}$ together define intervals for vertices located in $B_{ij}$.
And, it is not necessary for all the intervals on any obstacle to have their corresponding bands on the same splitter in $X_{ij}$.
The two lemmas below help compute the total number of intervals associated with any source vertex in $A_{ij}$.

\begin{lemma}
\label{lem:numinter}
The total number of intervals induced due to any source vertex $v \in \cal P$ is $O(h)$.
\end{lemma}
\begin{proof}
The intervals induced by $v$ on any obstacle can either be (i) sandwiched between two other intervals while the splitters corresponding to these three intervals are pairwise distinct, or (ii) the interval is not sandwiched between two other intervals.
Due to Lemma~5.2 in Banyassady~\etal~\cite{journals/cgta/BanyaCKMR20}, we know that there exists at most one obstacle in $\cal P$ that has intervals defined on its boundary due to any three pairwise distinct splitters.
Hence, the number of intervals of type~(i) is upper bounded by the total number of splitters in the separator components.
From Lemma~\ref{lem:numinter}, the latter is $O(\sqrt{h}\lg{h})$.
Since there can be at most two intervals on any obstacle that are not sandwiched between two other intervals, the number of intervals of type~(ii) is upper bounded by $O(h)$.
\end{proof}
Next, we describe our algorithm to compute $V_{ijk}$.
To remind, the set $V_{ijk}$ comprises every vertex $w \in B_{ij}$ such that the shortest path from $v$ to $w$ first intersects $\ell_{ijk}$ among all the splitters in $X_{ij}$. 
We use a plane sweep algorithm to compute these intersections, which we describe below.
Let $E_{ij}$ be the subset of edges of the shortest path tree rooted at $v$ that intersect component $S_{ij}$.
The line segments in $E_{ij} \cup X_{ij}$ intersected by the sweep line are stored at the leaves of a balanced binary search tree $\tau$. 
The top-to-bottom order of line segments that intersect the sweep line corresponds to the left-to-right ordering of leaves of $\tau$. 
The event points of the sweep line include the endpoints of splitters in $X_{ij}$ and the endpoints of shortest path tree edges.
With each such event point, we store the line segment incident to that point.
With each event point, we also store whether the line segment incident to that event point is a tree edge or a splitter. 
Let $p'$ be the event point dequeued.
If $p'$ is on an edge $e \in T_v$ and $p'$ is the right endpoint of edge $e$, then our algorithm deletes the segment $e$ from $\tau$. 
Analogously, if $p'$ is the left endpoint of $e \in T_v$, then we include the segment $e$ into $\tau$.
If $p'$ is the endpoint of a splitter $\ell_{ijk} \in X_{ij}$, then we iterate through the segments stored in the tree $\tau$: for every segment $\ell$, if $\ell$ is an edge $e'$ in the tree $T_v$, we check if $\ell$ intersects $\ell_{ijk}$. 
Suppose $e'$ intersects the splitter $\ell_{ijk} \in X_{ij}$.  
Let $u$ be the vertex of $e'$ which lies to the right of $\ell_{ijk}$.
Then, starting at $u$, we perform the depth-first traversal of the subtree $T_u$ rooted at $u$.
However, in this traversal, if a vertex $v' \in S_{ij} \cup B_{ij}$ has been traversed previously, we do not traverse the subtree $T_{v'}$.
If a newly visited vertex $v''$ lies in $B_{ij}$, then $v''$ is added to the set $V_{ijk}$.
At the end, we delete $e'$ from $\tau$.

\begin{lemma}
\label{lem:timecompint}
Computing all the intervals corresponding to all the vertices in $\cal P$, considering all possible splitters at all levels of the recursion tree, together takes $O(n^2 \lg{n})$ time.
\end{lemma}
\begin{proof}
Since two event points are added for every line segment in $E_{ij} \cup X_{ij}$ to the event queue and since these are the only event points, there are at most $O(|E_{ij}| +|X_{ij}|)$ event points enqueued into $Q$.
Hence, enqueuing and dequeuing these event points into $Q$ together take $O((|E_{ij}| + |X_{ij}|) \lg({|E_{ij}| + |X_{ij}|}))$ time.
Since there are $O(|E_{ij}| + |X_{ij}|)$ event points, it takes $O((|E_{ij}| + |X_{ij}|)$ $\lg ({|E_{ij}| + |X_{ij}|}))$ time for inserting (resp. deleting) all these line segments to (resp. from) $\tau$. 
Since $\sum_{i,j} (|E_{ij}| + |X_{ij}|) = O(n)$, the total time taken for all the operations on $Q$ and $\tau$ together is $O(n \lg {n})$.
For every tree edge $e \in E_{ij}$, since we only compute intersection with the first splitter and since finding the point of intersection of two line segments takes constant time, the total time taken to compute intersections is $O(n)$.
For an obstacle $P_r$, both computing $V_{ijk}^r$ and computing the corresponding intervals together take linear time.
Thus, including all these time complexities, it takes $O(n \lg {n})$ time to compute all the intervals induced by $v$. 
Considering there are $n$ vertices in $\cal P$, the time to compute all the intervals corresponding to all the vertices is $O(n^2 \lg{n})$.
\end{proof}

\subsection{\bf Computing Routing Tables}

In this section, we describe the construction of routing tables at every vertex $v \in A_{ij}$ using the centers and intervals computed with the algorithms described in previous subsections.

We start with describing our algorithm for computing entries of $\rho_v^1$ for every vertex $v \in A_{ij}$ that are resulting from processing vertices in $S_{ij}$ and $B_{ij}$.
Consider the set $I$ of intervals on obstacle $P_r$ corresponding to source vertex $v \in A_{ij}$.
For every interval $I = (id(v_s), id(v_e), id(\ell_{ijk}))$ induced on $P_r$ due to $v$, we store $I$ in an entry of $\rho_v^1$.
As mentioned, the entries in $I$ respectively denote the vertex $v_s$ at which this interval starts on the $bd(P_r)$, vertex $v_e$ at which this interval ends on the $bd(P_r)$, and the first splitter $\ell_{ijk}$ in $X_{ij}$ intersected by a shortest path from $v$ to any of the vertices belonging to this interval. 
Using $id(t)$, we search entries in $\rho_v^1$ and find the entry that contains the interval to which $t$ belongs.
From that entry of $\rho_v^1$, we find $id(\ell_{ijk})$.
By searching for $id(\ell_{ijk})$ in $L(t)$, we determine the center vertex $c$.
This says $c$ is the center vertex of cluster $C$ such that $t$'s geodesic projection on $\ell_{ijk}$ belongs to $C$.
To remind, the packet header stores the routing label $L(t)$ of $t$, and $L(t)$ consists of ordered tuples of splitter and center vertex identifiers. 
We also set $c$ as the pseudo-destination in the packet header.
Once the center vertex $c$ is found, we use the routing table $\rho_v^2$ to find the next hop in routing the packet to $c$. 
Essentially, when any interval $I$ with respect to source vertex $v$ is computed, the routing table $\rho_v^1$ at $v$ is populated with $I$.

Let $C_{ij}$ be the set comprising all the center vertices of all the splitters in $X_{ij}$.
For every center vertex  $c \in C_{ij}$, let  $T_c$ be the shortest path tree rooted at $c$.
An entry of $\rho_v^2$ consists of $v$'s parent in $T_c$ and the $id(c)$.
At every intermediate vertex $w$ that occurs after $v$ before the packet reaches $c$, the entries in $\rho_w^2$ help in forwarding the packet from $w$ to its next hop.

In the rest of this subsection, we detail our algorithms for computing the routing table $\rho_v^3$ at $v$.
Noting $\cal F$ is the set comprising simple polygons, wherein each polygon in $\cal F$ is associated with a unique leaf node of the recursion tree and each leaf node of the recursion tree is associated with a unique simple polygon in $\cal F$.
After being recursively forwarded, possibly via multiple centers, the packet could reach a vertex $v$ in a simple polygon $P \in {\cal F}$ such that both $v$ and $t$ are vertices of $P$.
Each polygon $P \in \cal{F}$ is a simple polygon bounded by exactly four polygonal chains: a left vertical chain, a right vertical chain, and two convex chains. 
The left and right chains consist of one or more splitters or the edges of the bounding box $B({\cal P})$.
And each of those two convex chains consists of obstacle edges or the sections of the top or bottom edges of $B({\cal P})$.
Since $P$ is associated with a leaf node of the recursion tree, there is no splitter located in the relative interior of $P$; essentially, $P$ is a face of the subdivision of $\cal P$ wherein that subdivision is due to splitters.
Hence, $P$ can only have at most two convex chains, each located on the boundary of a distinct obstacle.
When no entry of $\rho_v^1$ has $t$, it indicates that both vertices $v$ and $t$ are in a simple polygon $P$ such that $P \in \cal F$.
The preprocessing algorithm computes routing tables so that a packet is routed from $v$ to $t$ without ever passing through any vertex exterior to $P$.
With respect to source vertex $v \in P$, we partition each convex chain located in $P$ into at most three {\it pieces}.
That is, any piece of a convex chain located in $P$ is defined with respect to a source vertex located in $P$.
Any such piece starts at a vertex, say $v_s$, ends at a vertex, say $v_e$, and the piece defined by these two vertices is denoted by $[v_s, v_e]$.
All the vertices that occur between $v_s$ and $v_e$, including $v_s$ and $v_e$, belong to piece $[v_s, v_e]$.
For any vertex $v$ of $P \in \cal F$, any entry of $\rho_v^3$ consists of three vertices.
Let an entry of $\rho_v^3$ consists $v', v'',$ and $w$.
That is, this entry corresponds to the piece $[v', v'']$.
If the routing algorithm determines that $t$ belongs to piece $[v', v'']$, then the routing algorithm forwards the packet to $w$.
This implicitly means $v$ and $w$ are mutually visible.

\begin{figure}[ht]
    \centering

    \begin{subfigure}[t]{0.45\textwidth}
        \centering
        \includegraphics[height=1.5in]{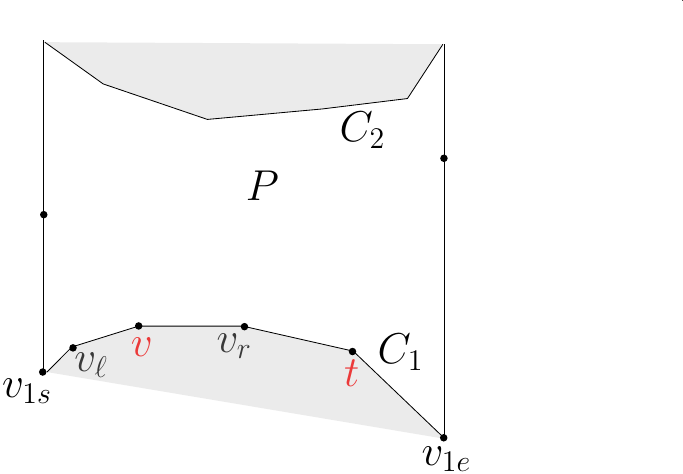}
        \caption{
            \footnotesize
            Illustrating the case when there is no tangent from a vertex $v$ on the chain $C_1$ to the chain $C_2$.
            \normalsize
        }
        \label{fig:no-tangent}
    \end{subfigure}
    \hfill
    \begin{subfigure}[t]{0.45\textwidth}
        \centering
        \includegraphics[height=1.2in]{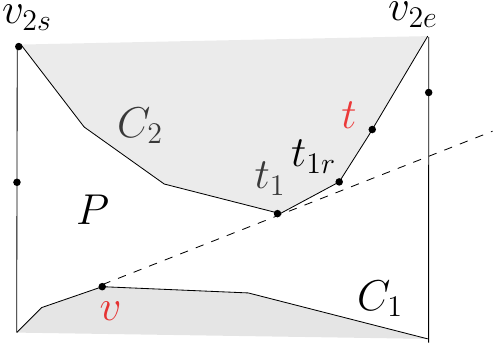}
        \caption{
            \footnotesize
            Illustrating the case when there is one tangent from a vertex $v$ on the chain $C_1$ to the chain $C_2$.
            \normalsize
        }
        \label{fig:one-tangent}
    \end{subfigure}

    \vspace{1em}

    \begin{subfigure}[t]{0.45\textwidth}
        \centering
        \includegraphics[height=2in]{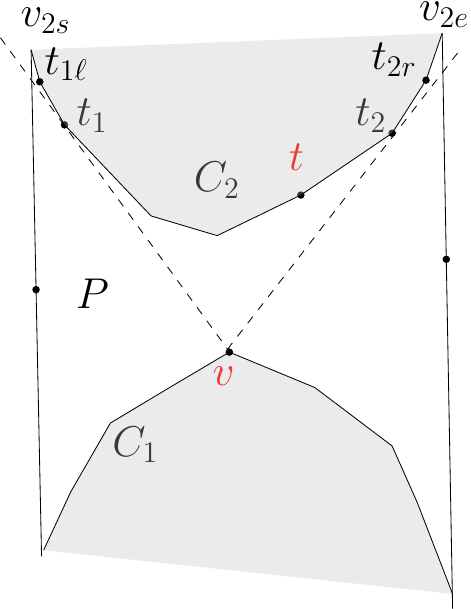}
        \caption{
            \footnotesize
            Illustrating the case when there are two tangents from a vertex $v$ on the chain $C_1$ to the chain $C_2$.
            \normalsize
        }
        \label{fig:two-tangents}
    \end{subfigure}
    \hfill
    \begin{subfigure}[t]{0.45\textwidth}
        \centering
        \includegraphics[height=2in]{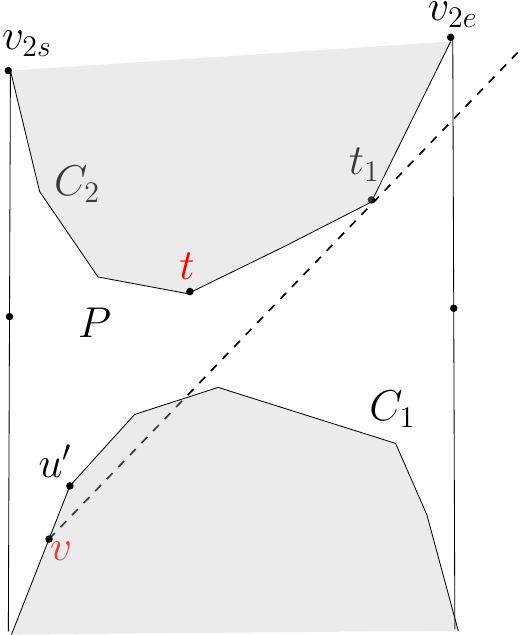}
        \caption{
            \footnotesize
            Illustrating the case when a tangent from a vertex $v$ on the chain $C_1$ to chain $C_1$ intersects $C_1$.
            \normalsize
        }
        \label{fig:intersecting-tangents}
    \end{subfigure}

    \caption{
        \footnotesize
        Illustrating the cases when both $v$ and $t$ lie within a polygon $P \in \cal F$.
        \normalsize
    }
    \label{fig:base-case}
\end{figure}
The boundary of $P$ consists of two convex chains, $C_1$ and $C_2$. 
Let $C_1$ (resp. $C_2$) be a section of the $bd(P_{r_1})$ (resp. $bd(P_{r_2})$).
Based on the locations of $v$ and $t$ on $C_1$ and $C_2$, the algorithm to compute $\rho_p^3$ is divided into four cases.
Without loss of generality, we suppose $v \in C_1$.
The $t$ could belong to either $C_1$ or $C_2$.
Let $v_{1s}$ and $v_{1e}$ respectively be the leftmost and rightmost vertices of $C_1$ in $P$.
Also, let $v_{2s}$ and $v_{2e}$ respectively be the leftmost and rightmost vertices of $C_2$ in $P$.

In Case~(a), assume $t \in C_1$.
Refer to Fig.~\ref{fig:no-tangent}.
Let $v_\ell$ (resp. $v_r$) be the first vertex to the left of $v$ (resp. right of $v$) on $C_1$.
If $t$ is located in the piece $[v_{1s},v_\ell]$, then packet is forwarded to $v_\ell$.
Corresponding entry in $\rho_v^3$ is $(id(v_{1s}),$ $id(v_\ell), id(v_\ell))$.
On the other hand, if $t$ is located in the piece $[v_{r}, v_{1e} ]$, then the packet is forwarded to $v_r$.
Corresponding entry in $\rho_v^3$ is $(id(v_r), id(v_{1e}), id( v_r))$.

Now let us consider the case in which $t \in C_2$.
Since $C_1$ and $C_2$ are convex chains forming the boundary of the simple polygon $P$, the visibility between $v$ and the vertices on $C_2$ can be characterized by drawing tangents from $v$ to $C_2$.
From a vertex $v$ on $C_1$, there could either be zero, one, or two tangents to $C_2$.
And, the tangent from $v$ could either intersect $C_1$ or not.
The tangents from $v$ to $C_2$ can be computed using a binary search over the vertices of $C_2$.
Analogously, the point of intersection of a tangent to $C_2$ from $v$ with $C_1$ can also be determined using a binary search over the vertices of $C_1$. 
The cases below arise from the number of tangents that can be drawn from $v$ to $C_2$.

Consider Case~(b) in which there is only one tangent from $v$ to $C_2$.
Let $t_1$ be the point of tangency of $v$ on $C_2$.
Refer to Fig.~\ref{fig:one-tangent}.
Let $t_{1r}$ be the first vertex to the right of $t_1$.
Without loss of generality, we assume that the vertices in the piece $[v_{2s}, t_1]$ are visible to $v$ and the vertices lying in the piece $[t_{1r},v_{2e}]$ are not visible from $v$.
When $t$ lies in the piece $[v_{2s},t_1]$, noting $t$ is visible to $v$, packet is forwarded to $t$.
Thus, the corresponding entry in $\rho_v^3$ is $(id(v_{2s}),id(t_1),id(t))$.
When $t$ lies in the piece $[t_{1r}, v_{2e}]$, noting $t_1$ is visible to $v$, the packet is forwarded to $t_1$.
The corresponding entry in $\rho_v^3$ is set to $(id(t_{1r}),id(v_{2e}),id(t_1))$. 

Case~(c) handles $C_2$ having two tangents from $v$ in $P$.
Let $t_1$ and $t_2$ be the corresponding points of tangencies on $C_2$ from $v$.
Refer to Fig.~\ref{fig:two-tangents}.
Without loss of generality, assume $t_1$ lies to the left of $t_2$.
Let $t_{1\ell}$ be the first vertex to the left of $t_1$ and let $t_{2r}$ be the first vertex to the right of $t_2$ respectively. 
First, we note that none of the vertices in either of the pieces $[v_{2s}, t_{1\ell}], [t_{2r}, v_{2e}]$ are visible from $v$.
If $t$ is located in the piece  $[v_{2s}, t_{1\ell}]$, packet is forwarded to $t_1$.
Hence, we set the corresponding entry in $\rho_v^3$ to be $(id(v_{2s}),id(t_{1\ell}),id(t_1))$.
If $t$ lies in the  piece $[t_1,t_2]$, packet is forwarded to $t$.
The corresponding entry in $\rho_v^3$ is $(id(t_1),id(t_2),id(t))$.
Lastly, if $t$ is located in the piece $[t_{2r},v_{2e}]$, packet is forwarded to $t_2$.
The corresponding entry in $\rho_v^3$ is $(id(t_{2r}),id(v_{2e}),id(t_2))$.

In Case~(d), we consider the tangent from $v$ to $C_2$ intersecting $C_1$.
Refer to Fig.~\ref{fig:intersecting-tangents}.
If the $x$-coordinate of the destination $t$ is greater than that of the current vertex $v$, we select the adjacent vertex to the right of $v$; otherwise, we choose the adjacent vertex to the left. 
This directional choice ensures that the packet moves closer to its destination.
Let $u'$ denote the selected adjacent vertex. 
We then store the entry $(id(v_{2s}), id(v_{2e}), id(u'))$ in the routing table $\rho_v^3$.

\begin{lemma}
\label{lem:routtabspace}
The space complexity of routing tables at any vertex of $\cal P$ is $O(h\lg{n} + \sqrt{h}(\lg{h})$ $(\min((\frac{1}{\epsilon})^{O( \lg {\alpha})},n))\lg {n})$, and the size of routing label of any vertex is $O(\sqrt{h} (\lg{h}) (\lg{n}))$.
Here, $\epsilon > 0$ and $\alpha > 1$.
\end{lemma}
\begin{proof}
For every vertex $v$, the total number of entries in $\rho_v^1$ is equal to the total number of intervals corresponding to $v$.
From Lemma~\ref{lem:numinter}, the latter is $O(h)$.
Since $h \le n$, the total number of entries in $\rho_v^1$ is $O(n)$.
For every vertex $v$, the total number of entries in $\rho_v^2$ is equal to the total number of center vertices.
From Lemma~{\ref{lem:numcenters}}, the total number of center vertices is $O(\sqrt{h}(\lg{h})(\min((\frac{1}{\epsilon})^{O(\lg {\alpha})},n)))$.
For any vertex $v$, since the routing information of at most eight pieces is stored in $\rho_v^3$, the number of entries in $\rho_v^3$ is $O(1)$.
Including all the entries in all three routing tables $\rho_v^1$, $\rho_v^2$ and $\rho_v^3$ at $v$, the total number of entries is $O(h+\sqrt{h}(\lg{h})(\min((\frac{1}{\epsilon})^{O(\lg {\alpha})},n)))$.
Considering there are $O(\lg{n})$ bits needed to store any entry, the space needed for routing tables at any vertex is $O(h\lg {n} + \sqrt{h}(\lg{h})(\min((\frac{1}{\epsilon})^{O(\lg {\alpha})},n))\lg {n})$.

Every vertex $v$ is projected on at most $O(\sqrt{h} \lg{h})$ splitters, and for every splitter, $v$'s projection belongs to a unique cluster. Thus, the number of entries in the routing label $L(v)$ is $O(\sqrt{h}\lg{h})$.
When multiplied by $\lg{n}$ bits required to store every entry, the space required to store $L(v)$ is $O(\sqrt{h}(\lg{h})(\lg{n}))$.
\end{proof}

\begin{lemma}
\label{lem:preproctime}
The preprocessing algorithm takes $O(n^2 \lg{n})$ time.
\end{lemma}
\begin{proof}
Given $\cal P$, computing the planar subdivision $\cal F$ takes $O(n +h\lg{h})$ time.
Since the number of levels in the recursion tree is $O(\lg{h})$, the recursive partitioning of the dual graph of $\cal F$ takes $O(h\lg{h})$ time.  
According to Lemma~\ref{lem:numsplitt}, $\sum_{i,j} |X_{ij}|= O(\sqrt{h} \lg{h})$. 
To sort splitters in $X_{ij}$, it takes $O(|X_{ij}| \lg {|X_{ij}|})$ time.
Thus, the time to compute the sorted orders of all $X_{ij}$s is $O(\sqrt{h} (\lg{h})^2)$.

Using the algorithm in Section~5.3 of de Berg~\etal~\cite{journals/jocg/deBergSKS24}, computing all the geodesic projections together takes $O(n^2\lg{n})$ time.
Using the algorithm in Hershberger and Suri~\cite{journals/siamcomp/HershbergerS99}, the time for computing all the shortest path trees, each respectively rooted at a vertex of $\cal P$, takes $O(n^2\lg {n})$ time. 
From Lemma~\ref{lem:timcompcen}, the time for computing all the center vertices of $\cal P$ together takes $O(n^2\lg {n})$.
From Lemma~\ref{lem:timecompint}, the time for computing all the intervals of all the vertices together takes $O(n^2 \lg{n})$ time.

From the proof of Lemma~\ref{lem:routtabspace}, the number of entries in $\rho_v^1$ is $O(n)$.
Since the size of every such entry is $O(\lg{n})$, the time for including all the entries to $\rho_v^1$ together takes $O(n\lg{n})$ time.
From the proof of Lemma~\ref{lem:routtabspace}, the number of entries in $\rho_v^2$ is $O(h+\sqrt{h}(\lg{h})(\min((\frac{1}{\epsilon})^{O(\lg {\alpha})},n)))$, which is $O(n)$.
Considering the size of each such entry is $O(\lg{n})$, the time taken to include all the entries to $\rho_v^2$ is $O(n\lg{n})$.
Again, from the proof of Lemma~\ref{lem:routtabspace}, the number of entries in $\rho_v^3$ is $O(1)$.
As part of computing $\rho_v^3$, finding tangents from $v$ to the convex chain takes $O(\lg{n})$ time.
Thus, the total time taken to compute all such tangents is $O(n \lg{n})$. 

Since the number of center vertices is $O(n)$ and the size of every entry is $\lg{n}$ bits, it takes $O(n\lg{n})$ time per vertex to compute the routing label for every vertex $v$.    
Considering $n$ vertices, the total time taken to add all routing table entries and compute all the routing labels is $O(n^2\lg{n})$.
\end{proof}

\section{Routing Algorithm}
\label{sect:routing}

Let $v$ be a vertex of $\cal P$ at which the packet resides.
Let $id(t), L(t), id(c)$ be the contents of the packet's header.
Here, $t$ is the destination vertex of the packet, $L(t)$ is the routing label of $t$, and $c$ is the pseudo-destination.
The routing algorithm first compares $id(v)$ with $id(t)$.
If they are equal, then the packet is determined to have reached its destination vertex.
Otherwise, if the pseudo-destination field is not empty and is not equal to $id(v)$, then the routing algorithm uses $\rho_v^2$ to find the next hop $h$ and forward the packet to $h$.
If $id(c)$ in the packet header is equal to $id(v)$, then the algorithm determines whether both $v$ and $t$ belong to the same polygon in $\cal F$.
This is done by searching for $t$ in the routing table $\rho_v^1$.
If the search fails, then $\rho_v^3$ is used to find the next hop $h$, and then the packet is forwarded to $h$.
Otherwise, $\rho_v^1$ is searched for an entry corresponding to $id(t)$. 
Once its entry in $\rho_v^1$ is found, the routing label $L(t)$ from the packet header is obtained.
Using the value of $id(\ell_{ijk})$ stored in that entry in $\rho_v^1$, we determine the center vertex $c$ from $L(t)$ and store $c$ in the packet header as the pseudo-destination of the packet.

\begin{lemma}
The multiplicative stretch of the routing path computed with this algorithm is $(7+\epsilon)\lg{h}$.
\end{lemma}
\begin{proof}
To remind, for any two vertices $v'$ and $v''$ of $\cal P$,  $d(v', v'')$ is the length of the shortest path in $\cal P$ between $v'$ and $v''$.
We let $\rho(v', v'')$ be the length of the (sub)path computed by the routing algorithm from $v'$ to $v''$.
Consider any two vertices $s$ and $t$ of $\cal P$.
It is immediate to note that $d(s, t) \le \rho(s, t)$.
Let $c_1, c_2, \ldots, c_k$ be the center vertices on the routing path computed by our algorithm from $s$ to $t$.
Then, since routing paths from $s$ to $c_1$, $c_1$ to $c_2$, \ldots, $c_{k-1}$ to $c_k$ are optimal, $\rho(s, t) = d(s, c_1) + d(c_1, c_2) + \ldots + \rho(c_k, t)$.
We consider the case when $s \in A_{ij}$ and $t \in B_{ij}$.  
Suppose the shortest path from $s$ to $t$ intersects splitter $\ell_{ijk}$.
And let $r$ be the point of intersection of this shortest path with $\ell_{ijk}$.
Then,
\begin{flalign*}
d(s,c_1) 
	&\le  d(c_1,t) + d(s,t)  && \\
	&\hspace{0.1in}\text{[by the triangle inequality]} && \\
	&\le d(c_1,c_{1\ell}) + |c_{1\ell}t_\ell| + d(t_\ell,t) + d(s,t) && \\
	&\hspace{0.1in}\text{[by the triangle inequality]} && \\
	&\le d(t,t_\ell) + |c_{1\ell}t_\ell| + d(t,t_\ell) + d(s,t) && \\
	&\hspace{0.1in}\text{[since $d(c_1 ,c_{1\ell}) \le d(t,t_\ell)$]} && \\
	&\le d(t,t_\ell) + \epsilon \hspace{0.02in} d(t,t_\ell) + d(t,t_\ell) + d(s,t) && \\
	&\hspace{0.1in}\text{[since $|c_{1\ell}t_\ell| \le \epsilon \hspace{0.02in} d(t ,t_\ell)$]} && \\
	&= (2+ \epsilon) d(t,t_\ell) + d(s,t). && \tag{1}
\end{flalign*}
Next, we upper bound the length of the routing path from $s$ to $c_k$.
Here, the point of geodesic projection of $t$ onto the splitter on which $c_i$ got projected is denoted by $t_{i\ell}$.
\begin{flalign*}
\rho(s, c_k) 
&= d(s, c_1) + d(c_1, c_2) + d(c_2, c_3) + \ldots + d(c_{k-1}, c_k) && \\
&\le (2 + \epsilon)  (d(t, t_{1\ell}) + d(t, t_{2\ell}) + \ldots + d(t, t_{k\ell})) + (\lg{h})  d(s, t) && \\
&\hspace{0.1in} \text{[from (1)]} && \\
&\le (\lg{h})  ( (2 + \epsilon)  d(t, t_{i\ell}) + (\lg{h})  d(s, t) ) && \\
&\hspace{0.1in} \text{[where $d(t, t_{i\ell})$ is the maximum among all $d(t, t_{j\ell})$]} && \\
&\le (\lg{h})  ( (2 + \epsilon)  d(t, r) + (\lg{h})  d(s, t) ) && \\
&\hspace{0.1in} \text{[since $d(t, t_{i\ell}) \le d(t, r)$]} && \\
&= (3 + \epsilon)(\lg{h})  d(s, t). && \tag{2}
\end{flalign*}
We now upper bound the routing path from $c_k$ to $t$.
\begin{flalign*}
\rho(c_k, t) 
&\le \rho(s, c_k) + d(s, t) &\\
&\hspace{0.1in}\text{[by the triangle inequality]} && \\
&\le (3 + \epsilon)(\lg{h}) d(s, t) + d(s, t) &\\
&\hspace{0.1in} \text{[from~(2)]} && \\
&= (4 + \epsilon)(\lg{h})\hspace{0.02in} d(s, t). &
\end{flalign*}
Thus, from (1) and (2), 
\begin{flalign*}
\rho(s, t) 
&= \rho(s, c_k) + \rho(c_k, t) &\\
&\le (3 + \epsilon)(\lg{h})\hspace{0.02in} d(s, t) + (4 + \epsilon)(\lg{h}) \hspace{0.02in} d(s, t) &\\
&= (7 + 2\epsilon)(\lg{h})\hspace{0.02in} d(s, t). &
\end{flalign*}
By substituting $\epsilon$ with $\epsilon/2$ in the algorithm, we obtain $(7 + \epsilon)(\lg{h})$ as the multiplicative stretch of the routing path.
\end{proof}

The following theorem summarizes the result.

\setcounter{theorem}{0}

\begin{theorem}
Given a convex polygonal domain $\cal P$ and an input parameter $\epsilon > 0$, the preprocessing algorithm assigns a unique label of size $O(\sqrt{h} (\lg{h}) \lg{n})$ to each vertex of $\cal P$ and it computes routing tables at every vertex of $\cal P$ of size $O(h\lg{n} + \sqrt{h}(\lg{h})(\min((\frac{1}{\epsilon})^{O(\lg{\alpha})},n))\lg{n})$ in $O(n^2(\lg{n}))$ time so that any packet is routed along a geodesic path with $(7 + \epsilon)(\lg{h})$ multiplicative stretch, while that packet header carries at most $2 \lg{n}$ bits of routing information. 
Here, $h$ is the number of convex polygonal obstacles in $\cal P$, $n$ is the number of vertices of $\cal P$, and $\alpha > 1$ is a geometric parameter.
\end{theorem}

\section{Conclusions}
\label{sect:conclu}

The algorithm presented in this paper devises a routing scheme for convex polygonal domains.
The preprocessing time, sizes of routing tables, and multiplicative stretch of this routing algorithm are functions of the number of vertices defining the polygonal domain, input parameter $\epsilon$, and a parameter $\alpha$ that depends on the geometry of the polygonal domain.
This is the first algorithm to preprocess the polygonal domain with divide-and-conquer. 
The divide-and-conquer approach, vertical line segments (splitters) to decompose the free space of the polygonal domain, geodesic projections of vertices onto splitters, clustering of points projected onto splitters, partitioning the boundaries of obstacles into intervals, and using several properties of geodesic shortest paths, helped devise the algorithm.
For a range of $\epsilon$ values, the algorithm proposed herein improves both the preprocessing time and the routing table size compared to Banyassady~\etal~\cite{journals/cgta/BanyaCKMR20}. 
When $\epsilon \ge 1$, both the size of routing tables and the preprocessing time of our algorithm are an improvement from the routing algorithm given for convex polygonal domains in Inkulu and Kumar~\cite{journals/ijfcs/InkuluKum24}.
Future work could include improving the routing parameters, preprocessing time, routing table sizes, and multiplicative stretch. 
It would be interesting to explore routing algorithms using primitives developed in this paper for simple polygonal domains. 

\subsection*{Acknowledgements}

This research of R. Inkulu is supported in part by the National Board for Higher Mathematics (NBHM) grant 2011/33/2023NBHM-R\&D-II/16198.

\bibliographystyle{plain}

\end{document}